\documentclass[11pt, a4paper]{article}

\usepackage{microtype}
\usepackage{enumerate}
\usepackage{todonotes}
\usepackage{graphicx} 
\usepackage{amsmath,amsfonts,amssymb, amsthm}
\usepackage{tikz}
\usepackage[ruled,vlined,linesnumbered]{algorithm2e}
\providecommand{\DontPrintSemicolon}{\dontprintsemicolon}

\usepackage{mathtools}
\usepackage{hyperref}
\usepackage{color}
\theoremstyle{plain}
\newtheorem{theorem}{Theorem}[section]
\newtheorem{corollary}[theorem]{Corollary}
\newtheorem{lemma}[theorem]{Lemma}

\newcommand{\ceil}[1]{\ensuremath{\left\lceil{#1}\right\rceil}}%
\newcommand{\ones}{\ensuremath{n_1}}%
\newcommand{\zeroes}{\ensuremath{n_0}}%

\newcommand{\hamm}{\ensuremath{\mathcal{H}}}%

\newcommand{\ProblemFormat}[1]{{\sc #1}}
\newcommand{\ProblemName}[1]{\ProblemFormat{#1}\xspace}

\newcommand{\probkkclique}[0]{\ProblemName{$k \times k$-Clique}}
\newcommand{\probTSAT}[0]{\ProblemName{$3$-SAT}}
\newcommand{\probMAV}[0]{\ProblemName{Minimax Approval Voting}}

\newcommand{\probCS}[0]{\ProblemName{Closest String}}

\newcommand{\poly}{{\mathrm{poly}}}
\newcommand{\OPT}{{\mathrm{OPT}}}

\newcommand{\prob}[1]{\mathrm{Pr}[#1]}
\newcommand{\probBigPar}[1]{\mathrm{Pr}\left[#1\right]}
\newcommand{\expected}[1]{\mathbb{E}[#1]}
\newcommand{\expectedBigPar}[1]{\mathbb{E}\left[#1\right]}

\newcommand{\Oh}{\ensuremath{\mathcal{O}}}
\newcommand{\Ohstar}{\ensuremath{\Oh^\star}}

\allowdisplaybreaks

\begin{document}

\title{Approximation and Parameterized Complexity\\ of Minimax Approval Voting}

\date{}

\author{
Marek Cygan\thanks{University of Warsaw, Warsaw, Poland, \texttt{\{cygan, kowalik, arkadiusz.socala\}@mimuw.edu.pl}. The work of M. Cygan is a part of the project TOTAL that has received funding from the European Research Council (ERC) under the European Union’s Horizon 2020 research and innovation programme (grant agreement No 677651).
 {\L}. Kowalik and A. Soca{\l}a are supported by the National Science Centre of Poland, grant number 2013/09/B/ST6/03136.}\hspace{-10pt}
\and
{\L}ukasz Kowalik\footnotemark[1]\hspace{-10pt}
\and
Arkadiusz Soca{\l}a\footnotemark[1]\hspace{-10pt}
\and
Krzysztof Sornat\thanks{University of Wroc{\l}aw, Wroc{\l}aw, Poland, \texttt{krzysztof.sornat@cs.uni.wroc.pl}. K. Sornat was supported by the National Science Centre, Poland, grant number 2015/17/N/ST6/03684. During the work on these results, Krzysztof Sornat was an intern at Warsaw Center of Mathematics and Computer Science.}
}



\maketitle

\begin{abstract}
We present three results on the complexity of \probMAV.
First, we study \probMAV parameterized by the Hamming distance $d$ from the solution to the votes. 
We show \probMAV admits no algorithm running in time $\Ohstar(2^{o(d\log d)})$, unless the Exponential Time Hypothesis (ETH) fails. 
This means that the $\Ohstar(d^{2d})$ algorithm of Misra et al. [AAMAS 2015] is essentially optimal.
Motivated by this, we then show a parameterized approximation scheme, running in time $\Ohstar(\left({3}/{\epsilon}\right)^{2d})$, which is essentially tight assuming ETH.
Finally, we get a new polynomial-time randomized approximation scheme for \probMAV, which runs in time $n^{\Oh(1/\epsilon^2 \cdot \log(1/\epsilon))} \cdot \poly(m)$, 
almost matching the running time of the fastest known PTAS for \probCS due to Ma and Sun [SIAM J. Comp. 2009].
\end{abstract}

\section{Introduction}


One of the central problems in artificial intelligence and computational social choice is aggregating preferences of individual agents (see the overview of Conitzer~\cite{Conitzer10}).
Here we focus on {\em multi-winner choice}, where the goal is to select a $k$-element subset of a set of candidates. 
Given preferences of the agents, the subset is identified by means of a voting rule. 
This scenario covers a variety od settings: nations elect members of parliament or societies elect committees~\cite{ChamberlinCourant}, web search engines choose pages to display in response to a query~\cite{DworkKNS01}, airlines select movies available on board~\cite{SkowronFL15}, companies select a group of products to promote~\cite{LuB11}, etc.

In this work we restrict our attention to the situation where each vote (expression of the preferences of an agent) is a {\em subset} of the candidates.
Various voting rules are studied. 
In the simplest one, Approval Voting (AV), occurences of each candidate are counted and $k$ most often chosen candidates are selected.
While this rule has many desirable properties in the single winner case~\cite{fishburn1978axioms}, in the multi-winner scenario its merits are often considered less clear~\cite{laslier2010handbook}.
Therefore, numerous alternative rules have been proposed (see~\cite{Kilgour2010}), including Satifaction Approval Voting (SAV, satifaction of an agent is the fraction of her approved candidates that are elected; the goal is to maximize the total satisfaction), Proportional Approval Voting (PAV: like SAV, but satisfaction of an agent whose $j$ approved candidates are selected is the $j$-th harmonic number $H_j$), Reweighted Approval Voting (RAV: a $k$-round scheme, in each round another candidate is selected).
In this paper we study a rule called Minimax Approval Voting (MAV), introduced by Brams, Kilgour, and Sanver~\cite{brams2007minimax}. 
Here, we see the votes and the choice as 0-1 strings of length $m$ (characteristic vertors of the subsets).
The goal is to minimize the maximum Hamming distance to a vote. 
(Recall that the Hamming distance $\hamm(x,y)$ of two strings $x$ and $y$ of the same length is the number of positions where $x$ and $y$ differ.)

Our focus is on the computational complexity of computing the choice based on the MAV rule. In the \probMAV decision problem, we are given a multiset $S=\{s_1,\ldots,s_n\}$ of $0$-$1$ strings of length $m$ (also called votes), and two integers $k$ and $d$. The question is whether there exists a string $s \in \{0,1\}^m$ with exactly $k$ ones such that for every $i=1,\ldots,n$ we have $\hamm(s,s_i) \le d$. 
In the optimization version of \probMAV we minimize $d$, i.e., given a multiset $S$ and an integer $k$ as before, the goal is to find a string $s \in \{0,1\}^m$ with exactly $k$ ones which minimizes $\max_{i=1,\ldots,n}\hamm(s,s_i)$. 

A reader familiar with string problems might recognize that \probMAV is tightly connected with the classical NP-complete problem called \probCS, where we are given $n$ strings over an alphabet $\Sigma$ and the goal is to find a string that minimizes the maximum Hamming distance to the given strings. Indeed, LeGrand~\cite{LeGrand04} showed that \probMAV is NP-complete as well by reduction from \probCS with binary alphabet. 
This motivated the study on \probMAV in terms of approximability and fixed-parameter tractability.

\paragraph{Previous results on \probMAV}

First approximation result was a simple 3-approximation algorithm due to LeGrand, Markakis and Mehta~\cite{LeGrandMM07}, obtained by choosing an arbitrary vote and taking any $k$ approved candidates from the vote (extending it arbitrarily to $k$ candidates if needed). 
Next, a 2-approximation was shown by Caragiannis, Kalaitzis and Markakis using an LP-rounding procedure~\cite{CaragiannisKM10}. 
Finally, recently Byrka and Sornat~\cite{ByrkaS14} presented a polynomial time approximation scheme (PTAS), i.e., an algorithm that for any fixed $\epsilon > 0$ gives a $(1+\epsilon)$-approximate solution in polynomial time. 
More precisely, their algorithm runs in time $m^{\Oh(1/\epsilon^4)} + n^{\Oh(1/\epsilon^3)}$ what is polynomial on number of voters $n$ and number of alternatives $m$. 
The PTAS uses information extraction techniques from fixed size ($\Oh(1/\epsilon)$) subsets of voters and random rounding of the optimal solution of a linear program.

In the area of fixed parameter tractability (FPT) the goal is to find algorithms with running time of the form $f(r)\poly(|I|)$, where $|I|$ is the size of the input istance $I$, $r$ is a parameter and $f$ is a function, which is typically at least exponential for NP-complete problems. For more about paremeterized algorithms see the textbook of Cygan et al.~\cite{CyganFKLMPPS15} or the survey of Bredereck et al.\cite{BredereckCFGNW14} (in the context of computational social choice).
The study of FPT algorithms for \probMAV was initiated by Misra, Nabeel and Singh~\cite{MisraNS15}. 
They show for example that \probMAV parameterized by the number of ones in the solution $k$ (i.e. $k$ is the paramater $r$) is $W[2]$-hard, which implies that there is no FPT algorithm, unless there is a highly unexpected collapse in parameterized complexity classes. From a positive perspective, they show that the problem is FPT when parameterized by the maximum allowed distance $d$.
Their algorithm runs in time\footnote{The $\Ohstar$ notation suppresses factors polynomial in the input size.} $\Ohstar(d^{2d})$\footnote{Actually, in the article~\cite{MisraNS15} the authors claim the slightly better running time of $\Ohstar(d^{d})$. However, it seems there is a flaw in the analysis: it states that the initial solution $v$ is at distance at most $d$ from the solution, while it can be at distance $2d$ because of what we call here the $k$-completion operation. This increases the maximum depth of the recursion to $d$ (instead of the claimed $d/2$).}.

\paragraph{Previous results on \probCS}

It is interesting to compare the known results on \probMAV with the corresponding ones on the better researched \probCS.
The first PTAS for \probCS was given by Li, Ma and Wang~\cite{LiMW02} with running time bounded by $n^{\Oh(1/\epsilon^4)}$.
This was later improved by Andoni et al.~\cite{AndoniIP06} to $n^{\Oh(\frac{\log 1/\epsilon}{\epsilon^2})}$, and then by Ma and Sun~\cite{MaS09} to $n^{\Oh(1/\epsilon^2)}$.

The first FPT algorithm for \probCS, running in time $\Ohstar(d^d)$ was given by Gramm, Niedermeier, and Rossmanith~\cite{GrammNR03}.
This was later improved by Ma and Sun~\cite{MaS09}, who gave an algorithm with running time $\Ohstar(2^{\Oh(d)}\cdot |\Sigma|^d)$, which is more efficient for constant-size alphabets. 
No further substantial progress is possible, since Lokshtanov, Marx and Saurabh~\cite{LokshtanovMS11} have shown that \probCS admits no algorithms in time $\Ohstar(2^{o(d\log d)})$ or $\Ohstar(2^{o(d\log |\Sigma|)})$ , unless the Exponential Time Hypothesis (ETH)~\cite{eth} fails.

The discrepancy between the state of the art for \probCS and \probMAV raises interesting questions. 
First, does the additional constraint in \probMAV really makes the problem harder and the PTAS has to be significantly slower?
Similarly, although in \probMAV the alphabet is binary, no $\Ohstar(2^{\Oh(d)})$-time algorithm is known, in contrary to \probCS. 
Can we find such an algorithm? 
The goal of this work is to answer these questions.


\paragraph{Our results}
We present three results on the complexity of \probMAV.
Let us recall that the Exponential Time Hypothesis (ETH) of Impagliazzo et al.~\cite{eth} states that there exists a constant $c > 0$, such that there is no algorithm solving \probTSAT in time $\Ohstar(2^{cn})$. During the recent years, ETH became the central conjecture used for proving tight bounds on the complexity of various problems, see~\cite{eth-survey} for a survey. 
We begin from showing that, unless the ETH fails, there is no algorithm for \probMAV running in time $\Ohstar(2^{o(d\log d)})$. 
In other words, the algorithm of Misra et al.~\cite{MisraNS15} is essentially optimal, and indeed, in this sense \probMAV is harder than \probCS.
Motivated by this, we then show a parameterized approximation scheme, i.e., a randomized Monte-Carlo algorithm which, given an instance $(S,k,d)$ and a number $\epsilon>0$, finds a solution at distance at most $(1+\epsilon)d$ in time $\Ohstar(\left({3}/{\epsilon}\right)^{2d})$ or reports that there is no solution at distance at most $d$.
Note that our lower bound implies that, under (randomized version of) ETH, this is essentially optimal, i.e., there is no parameterized approximation scheme running in time $\Ohstar(2^{o(d\log (1/\epsilon))})$. 
Indeed, if such an algorithm existed, by picking $\epsilon = 1/(d+1)$ we get an exact algortihm which contradicts our lower bound.
Finally, we get a new polynomial-time randomized approximation scheme for \probMAV, which runs in time $n^{\Oh(1/\epsilon^2 \cdot \log(1/\epsilon))} \cdot \poly(m)$.
Thus the running time almost matches the one of the fastest known PTAS for \probCS (up to a $\log(1/\epsilon)$ factor in the exponent).

\paragraph{Organization of the paper} 
In Section~\ref{sec:definitions} we introduce some notation and we recall standard probabability bounds that are used later in the paper. 
In Section~\ref{sec:lowerbound} we present our lower bound for \probMAV parameterized by $d$.
Next, in Section~\ref{sec:param-aprox-sch} we show a parameterized approximation scheme.
Finally, in Section~\ref{sec:newptas} we show a new randomized PTAS. 
The paper concludes with Section~\ref{sec:concludingremarks}, where we discuss directions for future work.

\section{Definitions and Preliminaries}\label{sec:definitions}

For every integer $n$ we denote $[n] = \{1,2,\dots,n\}$.
For a set of words $S\subseteq\{0,1\}^m$ and a word $x\in\{0,1\}^m$ we denote $\hamm(x,S)=\max_{s\in S}\hamm(x,s)$.
For a string $s \in \{0,1\}^m$, the number of $1$'s in $s$ is denoted as $\ones(s)$ and it is also called the Hamming weight of $s$; similarly $\zeroes(s)=m-\ones(s)$ denotes the number of zeroes.
Moreover, the set of all strings of length $m$ with $k$ ones is denoted by $S_{k,m}$, i.e., $S_{k,m} = \{s \in \{0,1\}^m: \ones(s) = k\}$. $s[j]$ means $j$-th letter of a string $s$. For a subset of positions $P \subseteq [m]$ we define a subsequence $s|_P$ by removing letters on positions $[m] \setminus P$ from $s$.

For a string $s\in\{0,1\}^m$, any string $s'\in S_{k,m}$ at distance $|\ones(s)-k|$ from $s$ is called a {\em $k$-completion} of $s$.
Note that it is easy to find such a $k$-completion $s'$: when $\ones(s)\ge k$ we obtain $s'$ by replacing arbitrary $\ones(s)-k$ ones in $s$ by zeroes; similarly when $\ones(s)< k$ we obtain $s'$ by replacing arbitrary $k-\ones(s)$ zeroes in $s$ by ones.

We will use the following standard Chernoff bounds (see e.g.Chapter 4.1 in \cite{MotwaniR95}).
\begin{theorem}\label{def:chernoffbounds}
 Let $X_1, X_2, \dots, X_n$ be $n$ independent random $0$-$1$ variables such that for every $i=1,\ldots,n$ we have $\probBigPar{X_i=1} = p_i$, for $p_i \in [0,1]$. Let $X = \sum_{i=1}^n X_i$. Then,
 \begin{itemize}
  \item  for any $0 < \epsilon \le 1$ we have:
  \begin{equation}\label{eq:chernoff01plus}
   \probBigPar{X > (1+\epsilon) \cdot \expectedBigPar{X}} \le \exp\left(-\tfrac{1}{3} \epsilon^2 \cdot \expectedBigPar{X} \right)
  \end{equation}
  \begin{equation}\label{eq:chernoff01minus}
   \probBigPar{X < (1-\epsilon) \cdot \expectedBigPar{X}} \le \exp\left(-\tfrac{1}{2} \epsilon^2 \cdot \expectedBigPar{X} \right)
  \end{equation}
  \item for any $1 < \epsilon$ we have:
  \begin{align}
   \probBigPar{X > (1+\epsilon) \cdot \expectedBigPar{X}} &\le \exp\left(-\tfrac{1}{3} \epsilon \cdot \expectedBigPar{X} \right) \label{eq:chernoff1plus} \\
   \probBigPar{X < (1-\epsilon) \cdot \expectedBigPar{X}} &= 0 \label{eq:chernoff1minus}
  \end{align}
  
 \end{itemize}
\end{theorem}

\section{A lower bound}
\label{sec:lowerbound}

In this section we show a lower bound for \probMAV parameterized by $d$. To this end, we use a reduction from a problem called \probkkclique.
In \probkkclique we are given a graph $G$ over the vertex set $V = [k] \times [k]$, i.e., $V$ forms a grid with $k$ rows and $k$ columns, and the question is whether in $G$ there is a clique containing exactly one vertex in each row.

\begin{lemma}
\label{lem:kk-red}
Given an instance $I = (G,k)$ of \probkkclique with $k\ge 2$, one can construct an instance $I' = (S, k, d)$ of \probMAV, such that $I'$ is a yes-instance iff $I$ is a yes-instance, $d=3k-3$ and the set $S$ contains $O(k^2 \binom{2k-2}{k-2})$ strings of length $k^2+2k-2$ each.
The construction takes time polynomial in the size of the output.
\end{lemma}

\begin{proof}
Each string in the set $S$ will be of size $m = k^2+2k-2$.
Let us split the set of positions $[m]$ into $k+1$ blocks, where the first 
$k$ blocks contain exactly $k$ positions each, and the last $(k+1)$-th block
contains the remaining $2k-2$ positions.
Our construction will enforce that if a solution exists, it will
have the following structure:
there will be a single $1$ in each of the first $k$ blocks and put all zeros in the last block.
Intuitively the position of the $1$ in the first block encodes the clique vertex of the first row of $G$,
the position of the $1$ in the second block encodes the clique vertex of the second row of, etc.

We construct the set $S$ as follows.
\begin{itemize}
  \item {\bf (nonedge strings)} For each pair of nonadjacent vertices $v, v' \in V(G)$ of $G$
  belonging to different rows, i.e., $v = (a,b)$, $v' = (a',b')$, $a \neq a'$, 
  we add to $S$ a string $s_{vv'}$, where all the blocks except $a$-th and $a'$-th
  are filled with zeros, while the blocks $a$, $a'$ are filled with ones, except
  the $b$-th position in block $a$ and the $b'$-th position in block $a'$ which are zeros (see Fig.~\ref{fig:nonedge-string}).
  Formally, $s_{vv'}$ contains ones at positions $\{(a-1)k + j : j \in [k], j \neq b\} \cup \{(a'-1)k + j : j \in [k], j \neq b'\}$.
  Note that the Hamming weight of $s_{vv'}$ equals $2k-2$.
  \item {\bf (row strings)} For each row $i \in [k]$ we create exactly $\binom{2k-2}{k-2}$ strings,
  i.e., for $i \in [k]$ and for each set $X$ of exactly $k-2$ positions in the $(k+1)$-th block
  we add to $S$ a string $s_{i,X}$ having ones at all positions of the $i$-th block
  and at $X$, all the remaining positions are filled with zeros (see Fig.~\ref{fig:row-string}).
  Note that similarly as for the nonedge strings the Hamming weight of each row string equals $2k-2$,
  and to achieve this property we use the $(k+1)$-th block.
\end{itemize}

\begin{figure}[t]
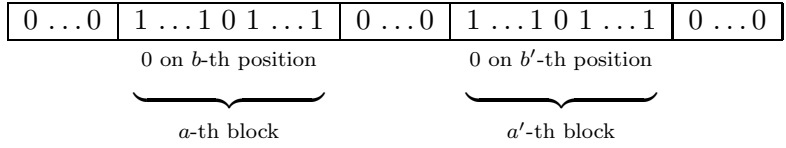

	\begin{center}
		\begin{tabular}{|c|c|c|c|c|}
			\hline
			0 \ldots 0           & 1 \ldots 1 0 1 \ldots 1               &
			0 \ldots 0           & 1 \ldots 1 0 1 \ldots 1               & 0 \ldots 0           \\ \hline
			\multicolumn{1}{c}{} & \multicolumn{1}{c}{\scriptsize $0$ on $b$-th position}  &
			\multicolumn{1}{c}{} & \multicolumn{1}{c}{\scriptsize $0$ on $b'$-th position} & \multicolumn{1}{c}{} \\
			\multicolumn{1}{c}{} & \multicolumn{1}{c}{\upbracefill}      &
			\multicolumn{1}{c}{} & \multicolumn{1}{c}{\upbracefill}      & \multicolumn{1}{c}{} \\
			\multicolumn{1}{c}{} & \multicolumn{1}{c}{\scriptsize $a$-th block}  &
			\multicolumn{1}{c}{} & \multicolumn{1}{c}{\scriptsize $a'$-th block} & \multicolumn{1}{c}{} \\
		\end{tabular}
	\end{center}
	\caption{Nonedge string.}
	\label{fig:nonedge-string}
\end{figure}

\begin{figure}[t]
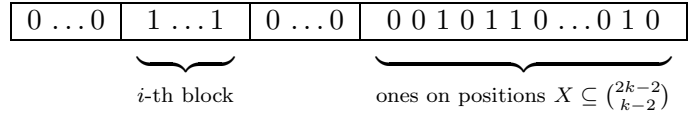

	\begin{center}
		\begin{tabular}{|c|c|c|c|}
			\hline
			0 \ldots 0           & 1 \ldots 1                                                             &
			0 \ldots 0           & 0 0 1 0 1 1 0 \ldots 0 1 0                                             \\ \hline
			\multicolumn{1}{c}{} & \multicolumn{1}{c}{\upbracefill}                                       &
			\multicolumn{1}{c}{} & \multicolumn{1}{c}{\upbracefill}                                       \\
			\multicolumn{1}{c}{} & \multicolumn{1}{c}{\scriptsize $i$-th block}                                   &
			\multicolumn{1}{c}{} & \multicolumn{1}{c}{\scriptsize ones on positions $X \subseteq {2k - 2 \choose k - 2}$}  \\
		\end{tabular}
	\end{center}
	\caption{Row string.}
	\label{fig:row-string}
\end{figure}

To finish the description of the created instance $I'=(S, k, d)$ we need to define the target distance $d$, which we set $d=3k-3$. Observe that as the Hamming weight of each string $s' \in S$ equals $2k-2$, for $s \in \{0,1\}^m$ with exactly $k$ ones we have $\hamm(s,s') \le d$ if and only if
the positions of ones in $s$ and $s'$ have a non-empty intersection.

Let us assume that there is a clique $K$ in $G$ of size $k$ containing
exactly one vertex from each row.
For $i \in [k]$ let $j_i \in [k]$ be the column number of the vertex of $K$
from row $i$.
Define $s$ as a string containing ones exactly at positions $\{(i-1)k + j_i : i \in [k]\}$,
i.e., the $(k+1)$-th block contains only zeros and for $i\in[k]$ the $i$-th
block contains a single $1$ at position $j_i$.
Obviously $s$ contains exactly $k$ ones, hence it suffices
to show that $s$ has at least one common one with each of the strings in $S$.
This is clear for the row strings, as each row string contains a block full of ones.
For a nonedge string $s_{vv'}$, where $v = (a,b)$ and $v'=(a',b')$ 
note that $K$ does not contain $v$ and $v'$ at the same time.
Consequently $s$ has a common one with $s_{vv'}$ in at least one of the blocks $a$, $a'$.

In the other direction, assume that $s$ is a string of length $m$ with
exactly $k$ ones such that the Hamming distance between $s$ and each 
of the strings in $S$ is at most $d$, which by construction implies that $s$
as a common one with each of the strings in $S$.
First, we are going to prove that $s$ contains a $1$
in each of the first $k$ blocks (and consequently has only zeros in block $k+1$).
For the sake of contradiction assume that this is not the case.
Consider a block $i \in [k]$ containing only zeros.
Let $X$ be any set of $k-2$ positions in block $k+1$ containing zeros from $s$
(such a set exists as block $k+1$ has $2k-2$ positions).
But the row string $s_{i,X}$ has $2k-2$ ones at positions where $s$ has zeros,
and consequently $\hamm(s,s_{i,X}) = k+(2k-2) = 3k-2 > d = 3k-3$, a contradiction.

As we know that $s$ contains exactly one one in each of the first $k$ blocks
let $j_i \in [k]$ be such a position of block $i \in [k]$.
Create $X \subseteq V(G)$ by taking the vertex from column $j_i$ for each row $i \in [k]$.
Clearly $X$ is of size $k$ and it contains exactly one vertex from each row, hence
it remains to prove that $X$ is a clique in $G$.
Assume the contrary and let $v, v' \in X$ be two distinct nonadjacent vertices of $X$,
where $v = (i,j_i)$ and $v' = (i',j_{i'})$.
Observe that the nonedge string $s_{vv'}$ contains zeros at the $j_i$-th position of the $i$-th
block and at the $j_{i'}$-th position of the $i'$-th block.
Since for $i'' \in [k]$, $i'' \neq i$, $i''\neq i$ block $i''$ of $s_{vv'}$ contains only
zeros, we infer that the sets of positions of ones of $s$ and $s_{vv'}$ are disjoint
leading to $\hamm(s,s_{vv'}) = k+(2k-2) = 3k-2 > d$, a contradiction.

As we have proved that $I$ is a yes-instance of \probkkclique iff $I'$ is a yes-instance of \probMAV,
the lemma follows.
\end{proof}

In order to derive an ETH-based lower bound we need the following theorem of Lokshtanov, Marx and Saurabh~\cite{LokshtanovMS11}.

\begin{theorem}
\label{thm:kxk-eth}
 Assuming ETH, there is no $2^{o(k \log k)}$-time algorithm for \probkkclique.
\end{theorem}

We are ready to prove the main result of this section.

\begin{theorem}
 Assuming ETH, there is no $2^{o(d \log d)}\poly(n,m)$-time algorithm for \probMAV.
\end{theorem}

\begin{proof}
 Using Lemma~\ref{lem:kk-red}, the input instance $G$ of \probkkclique is transformed into an equivalent instance $I'=(S,k,d)$ of \probMAV, where $n = |S| = O(k^2 \binom{2k-2}{k-2}) = 2^{O(k)}$, each string of $S$ has length $m=O(k^2)$ and $d=\Theta(k)$. 
 It follows that a $2^{o(d \log d)}\poly(n,m)$-time algorithm for \probMAV solves \probkkclique in time $2^{o(k \log k)} 2^{O(k)}=2^{o(k \log k)}$, which contradicts ETH by Theorem~\ref{thm:kxk-eth}.
 \end{proof}

%

\section{Parameterized approximation scheme}
\label{sec:param-aprox-sch}

In this section we show the following theorem. 

\begin{theorem}\label{thm:ptasForSmallOPT}
 There exists a randomized algorithm which, given an instance $(\{s_i\}_{i=1,\ldots,n},k,d)$ of \probMAV and any $\epsilon \in(0,3)$, runs in time $\mathcal{O}\left(\left(\frac{3}{\epsilon}\right)^{2d}mn\right)$ and either
 \begin{enumerate}[$(i)$]
  \item reports a solution at distance at most $(1+\epsilon)d$ from $S$, or
  \item reports that there is no solution at distance at most $d$ from $S$.
 \end{enumerate}
 In the latter case, the answer is correct with probabability at least $1-p$, for arbitrarily small fixed $p>0$.
\end{theorem}

Let us proceed with the proof. 
In what follows we assume $p=1/2$, since then we can get the claim even if $p<1/2$ by repeating the whole algorithm $\ceil{\log_2 (1/p)}$ times. 
Indeed, then the algorithm returns incorrect answer only if each of the $\ceil{\log_2 (1/p)}$ repetitions returned incorrect answer, which happens with probabability at most $(1/2)^{\log_2 (1/p)}=p$.

Assume we are given a yes-instance and let us fix a solution $s^*\in S_{k,m}$, i.e., a string at distance at most $d$ from all the input strings.
Our approch is to begin with a string $x_0\in S_{k,m}$ not very far from $s^*$, and next perform a number of steps.
In $j$-th step we either conclude that $x_{j-1}$ is already a $(1+\epsilon)$-approximate solution, or with some probability we find another string $x_j$ which is closer to $s^*$.

First observe that if $|\ones(s_1)-k|>d$, then clearly there is no solution and our algorithm reports NO.
Hence in what follows we assume 
\begin{equation}
\label{eq:close}
 |\ones(s_1)-k|\le d.
\end{equation}
 We set $x_0$ to be any $k$-completion of $s_1$.
 By~\eqref{eq:close} we get $\hamm(x_0,s_1) \le d$. 
 Since $\hamm(s_1,s^*)\le d$, by the triangle inequality we get the following bound.
\begin{equation}
\label{eq:x0}
 \hamm(x_0,s^*) \le \hamm(x_0,s_1) + \hamm(s_1,s^*) \le 2d.
\end{equation}

 \begin{algorithm}[t]
  \SetAlgorithmName{Pseudocode}{}
  \DontPrintSemicolon 
  {\small
   \lIf{$|\ones(s_1)-k|>d$}{\Return NO}
   $x_0 \gets$\ any $k$-completion of $s_1$\;\label{alg1:completion}
   \For{$j \in \{1,2,\dots,d\}$}{
    \lIf{ $\hamm(x_{j-1},S) \le (1+\epsilon)d$ }{\Return $x_{j-1}$}\label{alg1:returnxj}
    otherwise there exists $s_i$ s.t. $\hamm(x_{j-1},s_i) > (1+\epsilon)d$\;\label{alg1:there-exists-si}
    $P_{j,0} \gets \{ a \in [m]: 0 = x_{j-1}[a] \neq s_i[a] = 1 \}$\;\label{alg1:p0}
    $P_{j,1} \gets \{ a \in [m]: 1 = x_{j-1}[a] \neq s_i[a] = 0 \}$\;\label{alg1:p1}
    \lIf{ $\min(|P_{j,0}|,|P_{j,1}|) = 0$ }{\Return NO}
    Get $x_{j}$ from $x_{j-1}$ by swapping $0$ and $1$ on random positions from $P_{j,0}$ and $P_{j,1}$)\;
   }
   \lIf{ $\hamm(x_d,S) \le (1+\epsilon)d$ }{\Return $x_d$}
   \lElse{ \Return NO }
  }
  \caption{Parameterized approximation scheme for \probMAV.}\label{fig:alg-fpt}
 \end{algorithm}
 
 Now we are ready to describe our algorithm precisely (see also Pseudocode~\ref{fig:alg-fpt}).
We begin with $x_0$ defined as above. Next for $j=1,\ldots,d$ we do the following.
If for every $i=1,\ldots,n$ we have $\hamm(x_{j-1},s_i) \le (1+\epsilon)d$ the algorithm terminates and returns $x_{j-1}$.
Otherwise, fix any $i=1,\ldots,n$ such that $\hamm(x_{j-1},s_i) > (1+\epsilon)d$.
Let  $P_{j,0} = \{ a \in [m]: 0 = x_{j-1}[a] \neq s_i[a] = 1 \}$ and $P_{j,1} = \{ a \in [m]: 1 = x_{j-1}[a] \neq s_i[a] = 0 \}$.
The algorithm samples a position $a_0 \in P_{j,0}$ and a position $a_1 \in P_{j,1}$. 
Then, $x_j$ is obtained from $x_{j-1}$ by swapping the $0$ at position $a_0$ with the $1$ at position $a_1$.
If the algorithm finishes without finding a solution, it reports NO. 

The following lemma is the key to get a lower bound on the probablity that the $x_j$'s get close to $s^*$.

 \begin{lemma}\label{lem:manygoodpositions}
 Let $x$ be a string in $S_{k,m}$ such that $\hamm(x,s_i) \ge (1+\epsilon)d$ for some $i=1,\ldots,n$.
 Let $s^* \in S_{k,m}$ be any solution, i.e., a string at distance at most $d$ from all the strings $s_i$, $i=1,\ldots,n$.
 Denote
  \[ P_0^* = \left\{ a \in [m]: 0 = x[a] \neq s_i[a] = s^*[a] = 1 \right\}, \]
  \[ P_1^* = \left\{ a \in [m]: 1 = x[a] \neq s_i[a] = s^*[a] = 0 \right\}. \]
  Then,
  \[ \min\left(\left|P_0^*\right|, \left|P_1^*\right|\right) \ge \frac{\epsilon d}{2}. \]
 \end{lemma}
 \begin{proof}
 
 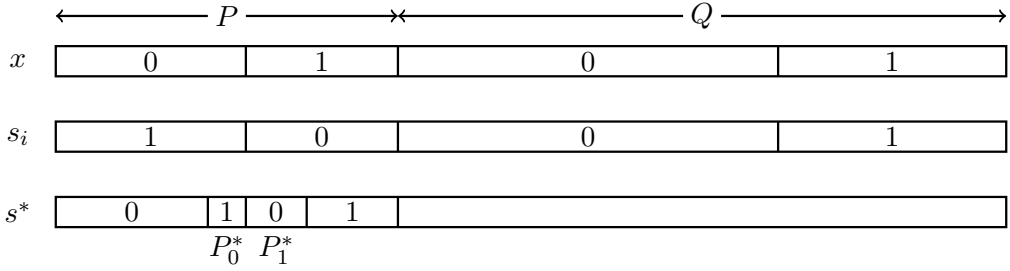
\begin{figure}[t]
 \centering
\begin{tikzpicture}[scale=1,line width=0.3mm]

\node at (1.75,2.8) {$P$};
\node at (8,2.8) {$Q$};
\node at (1.75,-0.3) {$P_0^*$};
\node at (2.4,-0.3) {$P_1^*$};

\draw[<-] (-0.5,2.8) -- (1.5,2.8);
\draw[->] (2,2.8) -- (4,2.8);
\draw[<-] (4,2.8) -- (7.75,2.8);
\draw[->] (8.25,2.8) -- (12,2.8);

\draw  (-0.5,0) rectangle (12,0.4);
\draw (2,0) -- (2,0.4);
\draw (4,0) -- (4,0.4);
\draw (1.5,0) -- (1.5,0.4);
\draw (2.8,0) -- (2.8,0.4);

\node at (-1,0.2) {$s^*$};

\node at (0.5,0.2) {$0$};
\node at (3.4,0.2) {$1$};
\node at (1.75,0.2) {$1$};
\node at (2.4,0.2) {$0$};

\draw  (-0.5,1) rectangle (12,1.4);
\draw (2,1) -- (2,1.4);
\draw (4,1) -- (4,1.4);
\draw (9,1) -- (9,1.4);

\node at (-1,1.2) {$s_i$};

\node at (0.75,1.2) {$1$};
\node at (3,1.2) {$0$};

\node at (6.5,1.2) {$0$};
\node at (10.5,1.2) {$1$};

\draw  (-0.5,2) rectangle (12,2.4);
\draw (2,2) -- (2,2.4);
\draw (4,2) -- (4,2.4);
\draw (9,2) -- (9,2.4);

\node at (-1,2.2) {$x$};

\node at (0.75,2.2) {$0$};
\node at (3,2.2) {$1$};

\node at (6.5,2.2) {$0$};
\node at (10.5,2.2) {$1$};

\end{tikzpicture}
\caption{\label{fig:3-strings}Strings $x$, $s_i$ and $s^*$ after permuting the letters.}
\end{figure}
 
 Let $P$ be the set of positions on which $x$ and $s_i$ differ, i.e., $P = \{ a \in [m]: x[a] \neq s_i[a] \}$. 
 (See Fig.~\ref{fig:3-strings}.)
 Note that $P_0^* \cup P_1^* \subseteq P$.
 Let $Q=[m]\setminus P$.
 
 The intuition behind the proof is that if $\min(|P_0^*|,|P_1^*|)$ is small, then $s^*$ differs too much from $s_i$, either because $s^*|_P$ is similar to $x|_P$ (when $|P_0^*|\approx |P_1^*|$) or because  $s^*|_Q$ has much more 1's than $s_i|_Q$ (when $|P_0^*|$ differs much from $|P_1^*|$).
 
 We begin with a couple of useful observations on the number of ones in different parts of $x$, $s_i$ and $s^*$.
 Since $x$ and $s_i$ are the same on $Q$, we get 
 \begin{equation}
 \label{eq:0}
 \ones(x|_Q)=\ones(s_i|_Q).
 \end{equation}
 Since $\ones(x)=\ones(s^*)$, we get $\ones(x|_P) + \ones(x|_Q) = \ones(s^*|_P) + \ones(s^*|_Q)$, and further 
\begin{equation}
 \label{eq:ie2}
  \ones(s^*|_Q) - \ones(x|_Q) = \ones(x|_P) - \ones(s^*|_P).
\end{equation}
 Finally note that 
 \begin{equation}
 \label{eq:ie3}
\ones(s^*|_P) = |P_0^*| + \ones(x|_P) - |P_1^*|.
\end{equation}

 We are going to derive a lower bound on $\hamm(s_i,s^*)$. First,
 \[
   \hamm(s_i|_P,s^*|_P) = |P| - (|P_0^*| + |P_1^*|) = \hamm(x,s_i) - (|P_0^*| + |P_1^*|) \ge (1+\epsilon)d - (|P_0^*| + |P_1^*|).
\]
 On the other hand,
\begin{align*}
 \hamm(s_i|_Q,s^*|_Q) & \ge |\ones(s^*|_Q) - \ones(s_i|_Q)| = \\
                      & \stackrel{\eqref{eq:0}}{=} |\ones(s^*|_Q) - \ones(x|_Q) = \\
                      & \stackrel{\eqref{eq:ie2}}{=} |\ones(x|_P) - \ones(s^*|_P)| = \\
                      & \stackrel{\eqref{eq:ie3}}{=} \left||P_1^*| - |P_0^*|\right|.\\
\end{align*}
It follows that
\begin{align*}
d \ge \hamm(s_i,s^*) = \hamm(s_i|_P,s^*|_P) + \hamm(s_i|_Q,s^*|_Q) & \ge (1+\epsilon)d - (|P_0^*| + |P_1^*|) + \left||P_1^*| - |P_0^*|\right| \\
                     & = (1+\epsilon)d - 2\min(|P_0^*|,|P_1^*|).
\end{align*}
Hence, $\min(|P_0^*|,|P_1^*|) \ge \frac{\epsilon d}{2}$ as required.
\end{proof}

\begin{corollary}
\label{cor:prob}
 Assume that there is a solution $s^*\in S_{k,m}$ and that the algorithm created a string $x_j$, for some $j=0,\ldots,d$.
 Then, 
 \[\prob{\hamm(x_j,s^*) \le 2d-2j} \ge \left(\frac{\epsilon}3 \right)^{2j}.\]
\end{corollary}

\begin{proof}
 We use induction on $j$. 
 For $j=0$ the claim follows from~\eqref{eq:x0}.
 Consider $j>0$. 
 By the induction hypothesis, 
 \begin{equation}
 \label{eq:pr1}
  \prob{\hamm(x_{j-1},s^*) \le 2d-2j+2} \ge \left(\frac{\epsilon}3 \right)^{2j-2}.
 \end{equation}
 Assume that $\hamm(x_{j-1},s^*) \le 2d-2j+2$.
 Since $x_j$ was created, $\hamm(x_{j-1},s_i) > (1+\epsilon)d$ for some $i=1,\ldots,n$.
 Since $\hamm(s^*,s_i) \le d$, by the triangle inequality we get the following.
 \begin{equation}
 \label{eq:pr2}
  |P_{j-1,0}| + |P_{j-1,1}| = \hamm(x_{j-1},s_i) \le \hamm(x_{j-1},s^*) + \hamm(s^*,s_i) \le 3d-2j+2 \le 3d. 
 \end{equation}
 Then, by Lemma~\ref{lem:manygoodpositions}
 \begin{equation}
 \label{eq:pr3}
 \prob{\hamm(x_{j},s^*) \le 2d-2j \ |\ \hamm(x_{j-1},s^*) \le 2d-2j+2} \ge \frac{|P_0^*|\cdot|P_1^*|}{|P_{j-1,0}| \cdot |P_{j-1,1}|}\ge \frac{\left(\frac{\epsilon d}{2}\right)^2}{\left(\frac{3d}{2}\right)^2} = \left(\frac{\epsilon}{3}\right)^2.
 \end{equation}
The claim follows by combining~\eqref{eq:pr1} and~\eqref{eq:pr3}.
\end{proof}

In order to increase the success probability, we repeat the algorithm until a solution is found or the number of repetitions is at least $\left({3}/{\epsilon}\right)^{2d}$.
By Corollary~\ref{cor:prob} the probablity that there is a solution but it was not found is bounded by
\[\left(1-\left(\frac{\epsilon}{3}\right)^{2d}\right)^{\left(3/\epsilon\right)^{2d}} = \left(1-\frac{1}{\left(3/\epsilon\right)^{2d}}\right)^{\left(3/\epsilon\right)^{2d}} \le e^{-1} < 1/2.\]
This finishes the proof of Theorem~\ref{thm:ptasForSmallOPT}.

\section{A fast polynomial time approximation scheme}
\label{sec:newptas}

The goal of this section is to present a PTAS for \probMAV running in time $n^{\mathcal{O}(1/\epsilon^2 \cdot \log(1/\epsilon))} \cdot \poly(m)$.
It is achieved by combining the parameterized approximation scheme from Theorem~\ref{thm:ptasForSmallOPT} with the following result, which might be of independent interest. 
Throughout this section $\OPT$ denotes the value of the optimum solution $s$ for the given instance $(\{s_i\}_{i=1,\ldots,n},k)$ of \probMAV, i.e., $\OPT=\max_{i=1,\ldots,n}\hamm(s,s_i)$,

\begin{theorem}\label{thm:ptasForBigOPT}
 There exists a randomized polynomial time algorithm which, for arbitrarily small fixed $p>0$, given an instance $(\{s_i\}_{i=1,\ldots,n},k)$ of \probMAV and any $\epsilon >0$ such that $\OPT\ge \frac{122 \ln n}{\epsilon^2}$, reports a solution, which with probabability at least $1-p$ is at distance at most $(1+\epsilon)\cdot\OPT$ from $S$.
\end{theorem}


In what follows, we prove Theorem~\ref{thm:ptasForBigOPT}.
As in the proof of Theorem~\ref{thm:ptasForSmallOPT} we assume w.l.o.g. $p=1/2$.
Note that we can assume $\epsilon <1$, for otherwise it suffices to use the 2-approximation of Caragiannis et al.~\cite{CaragiannisKM10}.
We also assume $n \ge 3$, for otherwise it is a straightforward exercise to find an optimal solution in linear time. 
 Let us define a linear program~(\ref{eq:ip_min}--\ref{eq:ip_xj_in_int01}):
 \begin{align}
  \text{minimize} \quad d &\label{eq:ip_min} &\\
  \sum_{j=1}^m x_j & = k\label{eq:ip_k_ones} &\\
  \sum_{\substack{j=1,\ldots,m \\ s_{i}[j]=1}} (1-x_j) + \sum_{\substack{j=1,\ldots,m \\ s_{i}[j]=0}} x_j & \le d & \forall i \in \{1,\dots,n\} \label{eq:ip_distance}\\
  x_j & \in  [0,1] & \forall j \in \{1,\dots,m\}\label{eq:ip_xj_in_int01}
 \end{align}

 The linear program~(\ref{eq:ip_min}--\ref{eq:ip_xj_in_int01}) is a relaxation of the natural integer program for \probMAV, obtained by replacing~\eqref{eq:ip_xj_in_int01} by the discrete constraint $x_j \in \{0,1\}$.
 Indeed, observe that $x_j$ corresponds to the $j$-th letter of the solution $x=x_1\cdots x_m$,~\eqref{eq:ip_k_ones} states that $\ones(x)=k$, and~\eqref{eq:ip_distance} states that $\hamm(x,S)\le d$. 
 
  \begin{algorithm}[t]
  \SetAlgorithmName{Pseudocode}{}
  \DontPrintSemicolon 
  {\small
   Solve the LP~(\ref{eq:ip_min}--\ref{eq:ip_xj_in_int01}) obtaining an optimal solution $(x_1^*,\ldots,x_m^*,d^*)$\;
   \For{$j \in \{1,2,\dots,m\}$}{
     Set $x[j] \gets 1$ with probability $x_j^*$ and $x[j] \gets 0$ with probability $1-x_j^*$ 
   }
   $y \gets$\ any $k$-completion of $x$\;\label{alg2:completion}
   \Return $y$
  }
  \caption{Parameterized approximation scheme for \probMAV}\label{alg:big-dist}
 \end{algorithm}
 
Our algorithm is as follows (see Pseudocode~\ref{alg:big-dist}).
 First we solve the linear program in time $\poly(n,m)$ using the interior point method \cite{Karmarkar84}. 
 Let $(x_1^*,\ldots,x_m^*,d^*)$ be the obtained optimal solution. 
 Clearly, $d^* \le \OPT$.
 We randomly construct a string $x\in\{0,1\}^m$, guided by the values $x_j^*$.
 More precisely, for every $j=1,\ldots,m$ independently, we set $x[j]=1$ with probabability $x_j^*$.
 Note that $x$ needs not contain $k$ ones.
 Let $y$ by any $k$-completion of $x$.
 The algorithm returns $y$.
 
 Clearly, the above algorithm runs in polynomial time.
 In what follows we bound the probability of error. 
 To this end we prove upper bounds on the probabability that $x$ is far from $S$ and the probabability that the number of ones in $x$ is far from $k$.
 This is done in Lemmas~\ref{lem:prob-x-far} and~\ref{lem:prob-ones}.
 
 \begin{lemma}
 \label{lem:prob-x-far}
 \[\probBigPar{\hamm(x,S) > (1+\tfrac{\epsilon}{2}) \cdot \OPT} \le \tfrac{1}4.\]
 \end{lemma}

 \begin{proof}
  For every $i=1,\ldots,n$ we define a random variable $D_i$ that measures the distance between $x^*$ and $s_i$
\[ D_i = \sum_{\substack{j \in [m]\\ s_i[j]=1}} (1-x[j]) + \sum_{\substack{j \in [m]\\ s_i[j]=0}} x[j]. \]
 Note that $x[i]$ are independent 0-1 random variables.

  Using linearity of the expectation we obtain
 \begin{align}
  \expected{D_i} &= \expectedBigPar{\sum_{j \in [m], s_i[j]=1} (1-x[j]) + \sum_{j \in [m], s_i[j]=0} x[j]} = \nonumber\\
  &= \sum_{j \in [m], s_i[j]=1} (1-\expected{x[j]}) + \sum_{j \in [m], s_i[j]=0} \expected{x[j]} = \nonumber\\
  &= \sum_{j \in [m], s_i[j]=1} (1-x_j^*) + \sum_{j \in [m], s_i[j]=0} x_j^* \le \nonumber\\
  & \le d^* \le \OPT.\label{ineq:ED_ileqd^*} 
 \end{align}
 Note that $D_i$ is a sum of $m$ independent 0-1 random variables $X_j = 1-x[j]$ when $s_i[j]=1$ and $X_j = x[j]$ otherwise.
 Denote $\delta = \epsilon \cdot \frac{\OPT}{2\expected{D_i}}$. 
 We apply Chernoff bounds.
 For $\delta < 1$ we have
 \begin{align*} &\prob{D_i > \left(1+\tfrac{\epsilon}{2}\right) \cdot \OPT} \stackrel{\eqref{ineq:ED_ileqd^*}}{\le} \probBigPar{D_i > \expected{D_i} + \tfrac{\epsilon}{2} \cdot \OPT} = \probBigPar{D_i > \left(1+\delta\right) \cdot \expected{D_i}} \stackrel{\eqref{eq:chernoff01plus}}{\le} \nonumber\\
 &\le \exp\left(-\frac{1}{3}\left(\epsilon \cdot \frac{\OPT}{2\expected{D_i}}\right)^2 \expected{D_i} \right) \stackrel{\eqref{ineq:ED_ileqd^*}}{\le} \exp\left(-\frac{\epsilon^2 \cdot \OPT}{12}\right).
 \end{align*}
 In case $\delta \ge 1$ we proceed analogously, using the Chernoff bound \eqref{eq:chernoff1plus}
  \begin{align*} &\prob{D_i > \left(1+\tfrac{\epsilon}{2}\right) \cdot \OPT} \stackrel{\eqref{eq:chernoff1plus}}{\le} \exp\left(-\frac{\epsilon \cdot \OPT}{6}\right) \stackrel{1 > \epsilon}{\le} \exp\left(-\frac{\epsilon^2 \cdot \OPT}{12}\right).
 \end{align*}
 Now we use the union bound to get the claim.
 \begin{align}
  \probBigPar{\hamm(x,S) \le (1+\tfrac{\epsilon}{2}) \cdot \OPT} &= \probBigPar{\exists i \in [n] \quad D_i > \left(1+\tfrac{\epsilon}{2}\right) \cdot \OPT} \le \nonumber\\
  &\le n \cdot \exp\left(-\frac{\epsilon^2 \cdot \OPT}{12}\right) \le \nonumber\\
  &\le n \cdot \exp\left(-\frac{\frac{122 \ln n}{\OPT} \cdot \OPT}{12}\right) < \nonumber\\
  &< n^{-9} \; \stackrel{n \ge 3}{<} \; \frac{1}{4}.
  \end{align}
  \end{proof}
 
 \begin{lemma}
 \label{lem:prob-ones}
  \[ \probBigPar{|\ones(x)-k| > \tfrac{\epsilon}{2} \cdot \OPT} < \frac{1}{4}.\]
 \end{lemma}
 
 \begin{proof}
  First we note that
  \begin{equation}\label{eq:EZeqk}
   \expected{\ones(x)} = \mathbb{E}\Big[\sum_{j \in [m]} x[j]\Big] = \sum_{j \in [m]} \expected{x[j]} = \sum_{j \in [m]} x_j^* \stackrel{\eqref{eq:ip_k_ones}}{=} k.  
  \end{equation}
  Pick an $i=1,\ldots,n$. Define the random variables
  \[ E_i =  \sum_{j \in [m], s_i[j]=1} (1-x[j]) , \quad F_i =  \sum_{j \in [m], s_i[j]=0} x[j].\]
  Let $D_i = E_i + F_i$, as in the proof of Lemma~\ref{lem:prob-x-far}. 
  By~\eqref{ineq:ED_ileqd^*}  we have
  \begin{equation}\label{ineq:EEileOPT}
    \expected{E_i} \le \expected{E_i} + \expected{F_i} = \expected{D_i} \le \OPT
  \end{equation}
  and analogously
  \begin{equation}
    \expected{F_i} \le \OPT.
  \end{equation}

  Both $E_i$ and $F_i$ are sums of independent 0-1 random variables and we apply Chernoff bounds as follows.  
  When $\frac{1}{4}\epsilon \cdot \frac{\OPT}{\expected{E_i}} \le 1$ then using \eqref{eq:chernoff01plus} and \eqref{eq:chernoff01minus} we obtain
  \begin{align*}
   &\probBigPar{ \Bigl| E_i - \expected{E_i} \Bigr| > \frac{1}{4}\epsilon \cdot \OPT } \stackrel{\eqref{eq:chernoff01plus}, \eqref{eq:chernoff01minus}}{\le} \\
   &\le \exp\left(-\frac{1}{3} \cdot \frac{1}{16}\epsilon^2 \cdot \frac{(\OPT)^2}{\mathbb{E}^2\left[E_i\right]} \cdot \expected{E_i} \right) + \exp\left(-\frac{1}{2} \cdot \frac{1}{16}\epsilon^2 \cdot \frac{(\OPT)^2}{\mathbb{E}^2\left[E_i\right]} \cdot \expected{E_i} \right) \stackrel{\ref{ineq:EEileOPT}}{\le} \\
   &\le 2 \cdot \exp\left(-\frac{1}{48} \epsilon^2 \cdot \OPT \right),
  \end{align*}
  otherwise $\left(\frac{1}{4}\epsilon \cdot \frac{\OPT}{\expected{E_i}} > 1\right)$, using \eqref{eq:chernoff1plus} and \eqref{eq:chernoff1minus}, we have
  \begin{align*}
   &\probBigPar{ \Bigl| E_i - \expected{E_i} \Bigr| > \frac{1}{4}\epsilon \cdot \OPT } \stackrel{\eqref{eq:chernoff1plus}, \eqref{eq:chernoff1minus}}{\le} \\
   &\le \exp\left(-\frac{1}{3} \cdot \frac{1}{4}\epsilon \cdot \frac{\OPT}{\expected{E_i}} \cdot \expected{E_i} \right) + 0 \le \\
   &\le \exp\left(-\frac{1}{12} \epsilon \cdot \OPT \right) \stackrel{1>\epsilon}{\le} 2 \cdot \exp\left(-\frac{1}{48} \epsilon^2 \cdot \OPT \right).
  \end{align*}
  To sum up, in both cases we have shown that
  \begin{equation}\label{eq:Ei_bound}
   \probBigPar{ \Bigl| E_i - \expected{E_i} \Bigr| > \frac{1}{4}\epsilon \cdot \OPT } \le 2 \cdot \exp\left(-\frac{1}{48} \epsilon^2 \cdot \OPT \right).
  \end{equation}
  Similarly we show
  \begin{equation}\label{eq:Fi_bound}
   \probBigPar{ \Bigl| F_i - \expected{F_i} \Bigr| > \frac{1}{4}\epsilon \cdot \OPT } \le 2 \cdot \exp\left(-\frac{1}{48} \epsilon^2 \cdot \OPT \right).
  \end{equation}
  We see that
  \begin{equation}\label{eq:Zeq}
   \ones(x) = \sum_{j \in [m]} x[j] = \ones(s_i) - \sum_{j \in [m], s_i[j]=1} (1-x[j]) + \sum_{j \in [m], s_i[j]=0} x[j] = \ones(s_i) - E_i + F_i  
  \end{equation}
  and hence
  \begin{equation}\label{eq:EZeq}
   \expected{\ones(x)} = \ones(s_i) - \expected{E_i} + \expected{F_i}.  
  \end{equation}
  Additionally we will use
  \begin{equation}\label{eq:xyineq}
   \forall x,y \in \mathbb{R} \quad |x-y| > a \implies |x| > a/2 \;\vee\; |y| > a/2.
  \end{equation}
  Now we can write
  \begin{align*}
   &\probBigPar{ \Bigl| \ones(x) - k \Bigr| > \tfrac{1}{2} \epsilon \cdot \OPT } \stackrel{\eqref{eq:EZeqk}}{=} \probBigPar{ \Bigl| \ones(x) - \expected{\ones(x)} \Bigr| > \tfrac{1}{2} \epsilon \cdot \OPT } \stackrel{\eqref{eq:Zeq},\eqref{eq:EZeq}}{=} \\
   &=  \probBigPar{ \Bigl| \ones(s_i) - E_i + F_i - \ones(s_i) + \expected{E_i} - \expected{F_i} \Bigr| > \tfrac{1}{2} \epsilon \cdot \OPT } \stackrel{\eqref{eq:xyineq}}{\le} \\
   &\le \probBigPar{ \Bigl| E_i - \expected{E_i} \Bigr| > \tfrac{1}{4} \epsilon \cdot \OPT \quad \vee \quad \Bigl| F_i - \expected{F_i} \Bigr| > \tfrac{1}{4} \epsilon \cdot \OPT } \le \\
   &\le \probBigPar{ \Bigl| E_i - \expected{E_i} \Bigr| > \tfrac{1}{4} \epsilon \cdot \OPT } + \probBigPar{ \Bigl| F_i - \expected{F_i} \Bigr| > \tfrac{1}{4} \epsilon \cdot \OPT } \stackrel{\eqref{eq:Ei_bound},\eqref{eq:Fi_bound}}{\le} \\
   &\le 4 \cdot \exp\left(-\tfrac{1}{48} \epsilon^2 \cdot \OPT \right) \stackrel{\text{assum.}}{\le} 4 \cdot \exp\left(-\tfrac{122}{48} \ln n \right) \stackrel{n \ge 3}{<} \tfrac{1}{4}.
  \end{align*}
 \end{proof}
 
 Now we can finish the proof of Theorem~\ref{thm:ptasForBigOPT}.
 By Lemmas~\ref{lem:prob-x-far} and~\ref{lem:prob-ones} with probabability at least $1/2$ both $\hamm(x,S) \le (1+\tfrac{1}{2}\epsilon) \cdot \OPT$ and $\hamm(y,x)=|\ones(x)-k| \le \tfrac{1}{2} \epsilon \cdot \OPT$. By triangle inequality this implies that $\hamm(y,S) \le (1+\epsilon) \cdot \OPT$, with probability at least $1/2$ as required.
 
 We conclude the section by combining Theorems~\ref{thm:ptasForSmallOPT} and~\ref{thm:ptasForBigOPT} to get a fast PTAS.

\begin{theorem}
 For each $\epsilon > 0$ we can find $(1+\epsilon)$-approximation solution for the \probMAV problem in time $n^{\mathcal{O}\left(\frac{\log{1/\epsilon}}{\epsilon^2} \right)} \cdot \poly(m)$ with probability at least $1-r$, for any fixed $r>0$.
\end{theorem}
\begin{proof}
 First we run algorithm from Theorem~\ref{thm:ptasForSmallOPT} for $d=\lceil\frac{122 \ln n}{\epsilon^2}\rceil$ and $p=r/2$.
 
 If it reports a solution, for every $d'\le d$ we apply Theorem~\ref{thm:ptasForSmallOPT} with $p=r/2$ and we return the best solution.
 If $\OPT \ge d$, even the initial solution is at distance at most $(1+\epsilon)d \le (1+\epsilon)\OPT$ from $S$.
 Otherwise, at some point $d'=\OPT$ and we get $(1+\epsilon)$-approximation with probability at least $1-r/2 > 1-r$.
 
 In the case when the initial run of the algorithm from Theorem~\ref{thm:ptasForSmallOPT} reports NO, we just apply the algorithm from Theorem~\ref{thm:ptasForBigOPT}, again with $p=r/2$. 
 With probability at least $1-r/2$ the answer NO of the algorithm from Theorem~\ref{thm:ptasForSmallOPT} is correct.
 Conditioned on that, we know that $\OPT > d \ge \frac{122 \ln n}{\epsilon^2}$ and then the algorithm from Theorem~\ref{thm:ptasForBigOPT} returns a $(1+\epsilon)$-approximation with probability at least $1-r/2$.
 Thus, the answer is correct with probabability at least $(1-r/2)^2 > 1-r$.
 
 The total running time can be bounded as follows.
 \[ \mathcal{O}^*\left(\left(\frac{3}{\epsilon}\right)^{ \frac{244 \ln n}{\epsilon^2}}\right) \subseteq \mathcal{O}^*\left( n^{ \mathcal{O}\left(\frac{\ln 1/\epsilon}{\epsilon^2}\right)}\right) \subseteq n^{ \mathcal{O}\left(\frac{\log 1/\epsilon}{\epsilon^2}\right)} \cdot \poly(m).\]
\end{proof}

\section{Further research}\label{sec:concludingremarks}
We conclude the paper with some questions related to this work that are left unanswered.
Our PTAS for \probMAV is randomized, and it seems there is no direct way of derandomizing it.
It might be interesting to find an equally fast deterministic PTAS.
The second question is whether there are even faster PTASes for \probCS or \probMAV.
Recently, Cygan, Lokshtanov, Pilipczuk, Pilipczuk and Saurabh~\cite{CyganLPPS16} showed that under ETH, there is no PTAS in time $f(\epsilon)\cdot n^{o(1/\epsilon)}$ for \probCS.
This extends to the same lower bound for \probMAV, since we can try all values $k=0,1,\ldots,m$.
It is a challenging open problem to close the gap in the running time of PTAS either for \probCS or for \probMAV.

\newpage



\end{document}